\documentclass[reqno, 12pt]{amsart}
\usepackage{amsmath, amsthm, amssymb, mathtools, graphicx, paralist, cite, color}
\usepackage[all]{xypic}
\definecolor{e-mail}{rgb}{0,.40,.80}
\definecolor{reference}{rgb}{.20,.60,.22}
\definecolor{citation}{rgb}{0,.40,.80}
\usepackage[pdfstartview=FitH, colorlinks, linkcolor=reference, citecolor=citation, urlcolor=e-mail]{hyperref}

\newtheorem{thm}{Theorem}
\newtheorem{cor}[thm]{Corollary}
\newtheorem{lem}[thm]{Lemma}
\newtheorem{prop}[thm]{Proposition}

\theoremstyle{definition}
\newtheorem{defn}[thm]{Definition}
\theoremstyle{remark}
\newtheorem{rem}[thm]{Remark}
\numberwithin{thm}{section}
\theoremstyle{definition}

\theoremstyle{definition}

\theoremstyle{definition}
\numberwithin{equation}{section}

\newcommand{\N}{\mathbb N} 
\newcommand{\Z}{\mathbb Z} 

\newcommand{\K}{\mathbb K} 
\newcommand{\Kx}{\mathbb{K}(x)}

\title{Computing discrete residues of rational functions}

\author{Carlos E. Arreche}
\address{Department of Mathematical Sciences \\ The University of Texas at Dallas}
\email{arreche@utdallas.edu}

\author{Hari P. Sitaula}
\address{Department of Mathematical Sciences \\  Montana Technological University}
\email{hsitaula@mtech.edu}

\usepackage{algorithm}
\usepackage{algpseudocode}

\begin{document}

\begin{abstract}
In 2012 Chen and Singer introduced the notion of discrete residues for rational functions as a complete obstruction to rational summability. More explicitly, for a given rational function $f(x)$, there exists a rational function $g(x)$ such that $f(x) = g(x+1) - g(x)$ if and only if every discrete residue of $f(x)$ is zero. Discrete residues have many important further applications beyond summability: to creative telescoping problems, thence to the determination of \mbox{(differential-)algebraic} relations among hypergeometric sequences, and subsequently to the computation of (differential) Galois groups of difference equations. However, the discrete residues of a rational function are defined in terms of its complete partial fraction decomposition, which makes their direct computation impractical due to the high complexity of completely factoring arbitrary denominator polynomials into linear factors. We develop a factorization-free algorithm to compute discrete residues of rational functions, relying only on gcd computations and linear algebra.
\end{abstract}


\maketitle

\tableofcontents

\section{Introduction} \label{sec:introduction}

Let $\K$ be field of characteristic zero, and consider the field $\Kx$ of rational functions in an indeterminate $x$ with coefficients in $\K$. First formulated in \cite{Abramov:1971}, the \emph{rational summation problem} asks, for a given $f(x)\in\Kx$, to construct $g(x),h(x)\in\Kx$ such that \begin{equation} \label{eq:summation-problem}f(x)=g(x+1)-g(x) +h(x)\end{equation} and the degree of the denominator of $h(x)$ is as small as possible. Such an $h(x)$ is called a \emph{reduced form} of $f(x)$.
The rational summation problem has a long and illustrious history \cite{Abramov:1971,Abramov:1975,Moenck:1977,Karr1981,Paule:1995,Pirastu1995a,Pirastu1995b, Malm:1995, Abramov1995,polyakov:2008}. It is clear that the problem admits a solution (by the well-ordering principle), and that such a solution is not unique, because for any solution $(g(x),h(x))$ we obtain another solution $(g(x)-h(x),h(x+1))$ since the degree of the denominator of $h(x+1)$ is the same as that of $h(x)$. In comparing the approaches in \emph{op.~cit.}, one can then ask for the denominator of $g(x)$ to be also as small as possible, and/or to compute some (any) solution $(g(x),h(x))$ to \eqref{eq:summation-problem} as efficiently as possible. We refer to the introduction of \cite{Pirastu1995b} for a concise summary and comparison between most of these different approaches.

Every algorithm for solving the rational summation problem also addresses, as a byproduct, the \emph{rational summability problem} of deciding, for a given $f(x)\in\Kx$, whether (just yes/no) there exists $g(x)\in \Kx$ such that $f(x)=g(x+1)-g(x)$, in which case we say $f(x)$ is \emph{rationally summable}. There are various algorithms for addressing this simpler question, designed to forego the usually expensive and often irrelevant computation of the \emph{certificate} $g(x)$, which are presented and discussed for example in \cite{Matusevich:2000,Gerhard:2003,BCCL:2010,Chen:2015,chen-singer:2012,Giesbrecht:2022} and the references therein.

We center our attention on the approach to rational summability proposed in \cite{chen-singer:2012}. The \emph{discrete residues} of $f(x)\in\Kx$ are constants defined in terms of the complete partial fraction decomposition of $f(x)$, and have the obstruction-theoretic property that they are all zero if and only if $f(x)$ is rationally summable. Computing these discrete residues directly from their definition is impractical due to the high computational cost of factoring the denominator of $f(x)$ into linear factors. We propose here algorithms for computing these discrete residues relying only on gcd computations and solving systems of linear equations in $\K$. To be clear, the discrete residue data of an arbitrary $f(x)$ are in general algebraic over $\K$. We submit that it would be perverse to avoid expensive factorizations throughout the algorithm, only to demand them at the very end! Inspired by \cite[Thm.~1]{bronstein:1993} our output consists of pairs of polynomials with coefficients in $\K$; one whose roots describe the places where $f(x)$ has non-zero discrete residues, the other whose evaluation at each such root gives the value of the corresponding discrete residue (see \S\ref{sec:dres-sum} for a more detailed description). Of course, any user who wishes to actually see the discrete residue data of $f$ may use the $\K$-polynomials produced by our algorithms to compute them explicitly to their heart's content and at their own risk.

Let us now describe our general strategy for computing discrete residues (cf.~Algorithm~\ref{alg:dres}). We apply iteratively Hermite reduction to $f(x)$ in order to reduce to the special case where the denominator of $f(x)$ is squarefree. Then we compute a reduced form $\bar{f}(x)$ of $f(x)$ whose denominator is both squarefree and shift-free, so that the discrete residues of $f(x)$ are the classical first-order residues of $\bar{f}(x)$. The factorization-free computation of the latter is finally achieved by \cite[Lem.~5.1]{trager:1976}.

Our proposed algorithms to compute discrete residues are obtained by combining in novel ways many old ingredients. Indeed, Hermite reduction is very old \cite{Ostrogradsky1845,Hermite1872}, and its iteration in Algorithm~\ref{alg:hermite-list} is already suggested in \cite[\S5]{Horowitz1971} for computing iterated integrals of rational functions. And yet, we have not seen this approach being more widely used in the literature, and it seems to us a good trick to have to hand. Indeed, we wonder whether it could provide a reasonable alternative, at least in some cases and for some purposes, to the algorithm in \cite{bronstein:1993} for symbolically computing complete partial fraction decompositions over the field of definition. Having thus reduced via Algorithm~\ref{alg:hermite-list} to the case where $f(x)$ has squarefree denominator, many of the varied earlier approaches to the summation and summability problems seem to accidentally collide into essentially the same procedure when restricted to this simpler situation. In this sense, our own reduction procedure described in \S\ref{sec:simple-reduction} strikes us as eerily similar to the one presented in \cite[\S5]{Gerhard:2003} over 20 years ago --- that ours may look simpler is a direct consequence of its being restricted to a simpler class of inputs. The simplicity of our approach allows us to exercise a great deal of control over the form of the outputs, in a ways which are particularly useful in developing some extensions of our basic procedures, elaborated in \S\ref{sec:extensions}. It is not obvious to us (but it would be interesting to see) how the same goals might be better accomplished differently, say by combining the reduction of \cite{Gerhard:2003} with the symbolic complete partial fraction decomposition algorithm of \cite{bronstein:1993}.

Our interest in computing discrete residues is motivated by the following variant of the summability problem, which often arises as a subproblem in algorithms for computing (differential) Galois groups associated with (shift) difference equations \cite{vanderput-singer:1997,hendriks:1998, HardouinSinger2008,arreche:2017}. Given several $f_1(x),\dots,f_n(x)\in\Kx$, compute (or decide non-existence) of $\mathbf{0}\neq\mathbf{v}=(v_1,\dots,v_n)\in\K^n$ such that \begin{equation}\label{eq:multi-telescoper}v_1f_1(x)+\dots+v_nf_n(x)=g_\mathbf{v}(x+1)-g_\mathbf{v}(x)\end{equation} for some $g_\mathbf{v}(x)\in\Kx$. Even if one wishes to compute the certificate $g_\mathbf{v}(x)$ explicitly, it is wasteful to perform a rational summation algorithm $n$ times for each $f_i(x)$ separately to produce $(g_i(x),h_i(x))$ as in \eqref{eq:summation-problem} as an intermediate step, because there is no guarantee that $h_\mathbf{v}(x)=\sum_iv_ih_i(x)$ has smallest possible denominator, so one may need to perform the algorithm an $(n+1)^\text{st}$ time to $h_\mathbf{v}(x)$, to decide summability. These inefficiencies are exacerbated in the more general context of \emph{creative telescoping problems} \cite{Zeilberger:1990,Zeilberger:1991,WZ:1992}, where the unknown $v_i\in\K$ are replaced with unknown linear differential operators $\mathcal{L}_i\in\K[\frac{d}{dx}]$. We refer to \cite{Chen:2019} for a succint and illuminating discussion of the history and computational aspects of creative telescoping problems, and in particular how the ``fourth generation'' reduction-based approaches bypass the computation of certificates, as our motivating problem \eqref{eq:multi-telescoper} illustrates.

Our approach based on discrete residues makes it very straightforward how to accommodate several $f_i(x)$ simultaneously as in \eqref{eq:multi-telescoper}, which adaptation is less obvious (to us) how to carry out efficiently using other reduction methods. On the other hand, we share in the reader's disappointment that we offer hardly any theoretical or experimental evidence supporting the efficiency of our approach in contrast with other possible alternatives. In fact, we expect that our approach will not be universally more efficient than some future adaptation of \cite[\S5]{Gerhard:2003} (for example) to the situation of \eqref{eq:multi-telescoper}, but rather that it can be a useful complement to it. We also expect that the conceptual simplicity of our approach will be useful in developing analogues to other related (but more technically challenging) contexts beyond the shift case, such as $q$-difference equations, Mahler difference equations, and elliptic difference equations, for which the corresponding notions of discrete residues have also been developed respectively in \cite{chen-singer:2012}, \cite{arreche-zhang:2022,arreche-zhang:2023}, and \cite{HardouinSinger2021}.

\section{Preliminaries} \label{sec:preliminaries}

\subsection{Basic notation and conventions} \label{sec:notation}

We denote by $\N$ the set of strictly positive integers, and by $\K$ a computable field of characteristic zero in which it is feasible to compute integer solutions to arbitrary polynomial equations with coefficients in $\K$. Such a field is termed \emph{canonical} in \cite{Abramov:1971}. We denote by $\overline\K$ a fixed algebraic closure of $\K$. We do not assume $\K$ is algebraically closed, and we will only refer to $\overline\K$ in proofs or for defining theoretical notions, never for computations.

We work in the field $\Kx$ of rational functions in a formal (transcendental) indeterminate $x$. For $f(x)\in \Kx$, we define \[
\sigma:f(x)\mapsto f(x+1);\qquad\text{and}\qquad \Delta:f(x)\mapsto f(x+1)-f(x).
\] Note that $\sigma$ is a $\K$-linear field automorphism of $\Kx$ and $\Delta=\sigma-\mathrm{id}$ is only a $\K$-linear map with $\mathrm{ker}(\Delta)=\K$. We often suppress the functional notation and write $f$ instead of $f(x)$, $\sigma(f)$ instead of $f(x+1)$, etc., when no confusion is likely to arise.

A \emph{proper} rational function is either $0$ or else has numberator of strictly smaller degree than that of the denominator. We assume implicitly throughout that rational functions are normalized to have monic denominator. Even when our rational functions are obtained as (intermediate) outputs of some procedures, we will take care to arrange things so that this normalization always holds. In particular, we also assume that the outputs of $\mathrm{gcd}$ and $\mathrm{lcm}$ procedures are also always normalized to be monic. An unadorned $\mathrm{gcd}$ or $\mathrm{lcm}$ or $\mathrm{deg}$ means that it is with respect to $x$. On the few occasions where we need a $\mathrm{gcd}$ with respect to a different variable $z$, we shall write $\mathrm{gcd}_z$. We write $\frac{d}{dx}$ (resp., $\frac{d}{dz}$) for the usual derivation operator with respect to $x$ (resp., with respect to $z$).

\subsection{Partial fraction decompositions}

A polynomial $b\in\Kx$ is \emph{squarefree} if $b\neq 0$ and $\mathrm{gcd}\left(b,\frac{d}{dx}b\right)=1$. Consider a proper rational function $f=\frac{a}{b}\in\Kx$, with $\mathrm{deg}(b)\geq 1$ and $\mathrm{gcd}(a,b)=1$. Suppose further that $b$ is squarefree, and that we are given a set $b_1,\dots,b_n\in\K[x]$ of monic non-constant polynomials such that $\mathrm{gcd}(b_i,b_j)=1$ whenever $i\neq j$ and $\prod_{i=1}^nb_i=b$. Then there exist unique non-zero polynomials $a_1,\dots,a_n\in\K[x]$ with $\mathrm{deg}(a_i)<\mathrm{deg}(b_i)$ for each $i=1,\dots,n$ such that $f=\sum_{i=1}^n\frac{a_i}{b_i}$. In this situation, we denote \begin{equation}\label{eq:parfrac-def} \mathtt{ParFrac}(f;b_1,\dots,b_n):=(a_1,\dots,a_n).    
\end{equation} We emphasize that the computation of partial fraction decompositions \eqref{eq:parfrac-def} can be done very efficiently \cite{kung:1977}, provided that the denominator $b$ of $f$ has already been factored into pairwise relatively prime factors $b_i$, which need not be irreducible in $\K[x]$. One can similarly carry out such partial fraction decompositions more generally for pre-factored denominators $b$ that are not necessarily squarefree. But here we only need to compute partial fraction decompositions for pre-factored squarefree denominators, in which case the notation \eqref{eq:parfrac-def} is conveniently light.

\subsection{Summability and dispersion} \label{sec:sum-disp}

We say $f\in\Kx$ is \emph{(rationally) summable} if there exists $g\in\Kx$ such that $f=\Delta(g)$. For a non-constant polynomial $b\in\K[x]$, we follow the original \cite{Abramov:1971} in defining the \emph{dispersion} of $b$ \[\mathrm{disp}(b):=\mathrm{max}\{\ell\in\N \ | \ \mathrm{gcd}(b,\sigma^\ell(b))\neq 1\}.\] For a reduced rational function $f=\frac{a}{b}\in\Kx$ with $\mathrm{gcd}(a,b)=1$ and $b\notin \K$, the \emph{polar dispersion} $\mathrm{pdisp}(f):=\mathrm{disp}(b)$.

We denote by $\overline\K/\Z$ the set of orbits for the action of the additive group $\Z$ on $\overline\K$. For $\alpha\in\overline\K$, we denote \[\omega(\alpha):=\{\alpha+n \ | \ n\in\Z\},\] the unique orbit in $\overline\K/\Z$ containing $\alpha$. We will often simply write $\omega\in\overline\K/\Z$ whenever there is no need to reference a specific $\alpha\in\omega$.

\section{Discrete residues and summability}\label{sec:dres-sum}

\begin{defn}[\protect{\cite[Def.~2.3]{chen-singer:2012}}]\label{def:dres} Let $f\in\Kx$, and consider the complete partial fraction decomposition
\begin{equation} \label{eq:partial-fraction-decomposition}f=p+\sum_{k\in\N}\sum_{\alpha\in\overline\K}\frac{c_k(\alpha)}{(x-\alpha)^k},\end{equation} where $p\in\K[x]$ and all but finitely many of the $c_k(\alpha)\in\overline\K$ are zero for $k\in\N$ and $\alpha\in\overline\K$. We define the \emph{discrete residue} of $f$ of order $k\in\N$ at the orbit $\omega\in\overline\K/\Z$ to be
\begin{equation}\label{eq:dres-def}
    \mathrm{dres}(f,\omega,k):=\sum_{\alpha\in\omega} c_k(\alpha).
\end{equation}
\end{defn}

The relevance of discrete residues to the study of rational summability is captured by the following result.

\begin{prop}[\protect{\cite[Prop.~2.5]{chen-singer:2012}}] \label{prop:dres-sum} 
$f\in\Kx$ is rationally summable if and only if $\mathrm{dres}(f,\omega,k)=0$ for every $\omega\in \overline\K$ and $k\in\N$.
    
\end{prop} As pointed out in \cite[Rem.~2.6]{chen-singer:2012}, the above Proposition~\ref{prop:dres-sum} recasts in terms of discrete residues a well-known rational summability criterion that reverberates throughout the literature, for example in \cite[p.~305]{Abramov1995}, \cite[Thm.~10]{Matusevich:2000}, \cite[Thm.~11]{Gerhard:2003}, \cite[Cor.~1]{ash:2005}. All of these rely in some form or another on the following fundamental result of Abramov, that already gives an important obstruction to summability.

\begin{prop}[\protect{\cite[Prop.~3]{Abramov:1971}}] \label{prop:summable-dispersion} If a proper rational function $0\neq f\in\Kx$ is rationally summable then $\mathrm{pdisp}(f)>0$.
\end{prop}

Discrete residues are also intimately related to the computation of reduced forms for $f$, in the sense that, as discussed in \cite[\S2.4]{chen-singer:2012}, every reduced form $h$ of $f$ as in \eqref{eq:summation-problem} has the form \begin{equation}\label{eq:general-reduced-form}h=\sum_{k\in\N}\sum_{\omega\in\overline\K}\frac{\mathrm{dres}(f,\omega,k)}{(x-\alpha_{\omega})^k}\end{equation} for some arbitrary choice of representatives $\alpha_{\omega}\in\omega$. Conversely, for every $h$ of the form \eqref{eq:general-reduced-form}, an immediate application of Proposition~\ref{prop:dres-sum} yields that $f-h$ is rationally summable. Another equivalent characterization for $h\in\Kx$ to be a reduced form is for it to have polar dispersion $0$. By \eqref{eq:general-reduced-form}, knowing the $\mathrm{dres}(f,\omega,k)$ is ``the same'' as knowing some/all reduced forms $h$ of $f$. But discrete residues still serve as a very useful organizing principle and technical tool, for both theoretical and practical computations.

For a given $f\in\K(x)$, our goal is to compute polynomials $B_k(x),D_k(x)\in\K[x]$ for each $k\in\N$ such that $(B_k,D_k)=(1,0)$ if and only if $\mathrm{dres}(f,\omega,k)=0$ for every $\omega\in\overline\K$ (which holds for all but finitely many $k\in\N$) and, for the remaining $k\in\N$, we have $0\leq\mathrm{deg}(D_k)<\mathrm{deg}(B_k)$ and $B_k$ is squarefree with $\mathrm{disp}(B_k)=0$. These polynomials will have the property that the set of roots $\alpha\in\overline\K$ of $B_k$ is a complete and irredundant set of representatives for all the orbits $\omega\in\overline\K/\Z$ such that $\mathrm{dres}(f,\omega,k)\neq 0$, and for each such root $\alpha\in\overline\K$ such that $B_k(\alpha)=0$, we have $\mathrm{dres}(f,\omega(\alpha),k)=D_k(\alpha)$.

\section{Iterated Hermite reduction} \label{sec:hermite}

It is immediate that the polynomial part of $f\in\Kx$ in \eqref{eq:partial-fraction-decomposition} is irrelevant, both for the study of summability as well as for the computation of discrete residues. So in this section we restrict our attention to proper rational functions $f\in\Kx$.

Our first task is to reduce to the case where $f$ has squarefree denominator. In this section we describe how to compute $f_k\in \Kx$ for $k\in\N$ such that, relative to the theoretical partial fraction decomposition \eqref{eq:partial-fraction-decomposition} of $f$, we have \begin{equation} \label{eq:hermite-list-output-k}
    f_k=\sum_{\alpha\in\overline\K}\frac{c_k(\alpha)}{x-\alpha}.
\end{equation} Of course we will then have by Definition~\ref{def:dres}\begin{equation}
    \label{eq:hermite-list-dres-k}\mathrm{dres}(f,\omega,k)=\mathrm{dres}(f_k,\omega,1)
\end{equation} for every $\omega\in\overline\K$ and $k\in\N$.

Our computation of the $f_k\in\Kx$ satisfying \eqref{eq:hermite-list-output-k} is based on iterating classical so-called \emph{Hermite reduction} algorithms, originally developed in \cite{Ostrogradsky1845,Hermite1872} and for which we refer to the fantastic modern reference \cite[\S2.2,\S2.3]{bronstein-book}.

\begin{defn} For proper $f\in\Kx$, the \emph{Hermite reduction} of \mbox{$f$ is} \[\mathtt{HermiteReduction}(f)=(g,h)\] where $g,h\in\Kx$ are proper rational functions such that \[f=\frac{d}{dx}(g)+h\] and $h$ has squarefree denominator.    
\end{defn}

The following Algorithm~\ref{alg:hermite-list} computes the $f_k\in\Kx$ satisfying \eqref{eq:hermite-list-output-k} by applying Hermite reduction iteratively and scaling the intermediate outputs appropriately.

\begin{algorithm}
\caption{$\mathtt{HermiteList}$ procedure}\label{alg:hermite-list}

\begin{algorithmic}

\Require A proper rational function $ 0\neq f\in\Kx$.
\Ensure A list $(f_1,\dots,f_m)$ of $f_k\in\Kx$ satisfying \eqref{eq:hermite-list-output-k}, such that $c_k(\alpha)=0$ for every $k>m$ and every $\alpha\in\overline\K$, with $f_m\neq 0$.
\\
\State Initialize loop: $m \gets 0$; $g \gets f$;

\While{$g \neq 0$}
\State {$(g, \hat{f}_{m+1}) \gets \mathtt{HermiteReduction}(g)$};
    \State $m \gets m+1$;
\EndWhile;
\State $f_k \gets (-1)^{k-1}(k-1)!\hat{f}_k$;\\
\Return $(f_1,\dots,f_m)$.
\end{algorithmic}
\end{algorithm}

\begin{lem} \label{lem:hermite-list}
    Algorithm~\ref{alg:hermite-list} is correct.
\end{lem}

\begin{proof}
    Letting $||f||= m\in\N$ denote the highest order of any pole of $ f\in\Kx$, note that $f_1,\dots,f_m\in\overline\K(x)$ defined by \eqref{eq:hermite-list-output-k} are uniquely determined by having squarefree denominator and satisfying \begin{equation}\label{eq:hermite-list-condition}f=\sum_{k=1}^m\frac{(-1)^{k-1}}{(k-1)!}\frac{d^{k-1}}{dx^{k-1}}(f_k).\end{equation} Defining inductively $g_0:=f$ and \begin{gather}\notag (g_{k},\hat{f}_{k}):=\mathtt{HermiteReduction}(g_{k-1}) \\ \Longleftrightarrow\qquad g_{k-1} = \frac{d}{dx}(g_{k})+\hat{f}_{k}\label{eq:hermite-list-induction}\end{gather} for $k\in\N$ as in Algorithm~\ref{alg:hermite-list}, we obtain by construction that all $g_k,\hat{f}_k\in\Kx$ and every $\hat{f}_k$ has squarefree denominator. Moreover, $||g_k||=||g_{k-1}||-1=m-k$, and therefore the algorithm terminates with $g_m=0$. Moreover, it follows from \eqref{eq:hermite-list-induction} that \[\sum_{k=1}^m\frac{d^{k-1}\hat{f}_k}{dx^{k-1}}=\sum_{k=1}^m\left(\frac{d^{k-1}g_{k-1}}{dx^{k-1}} - \frac{d^{k}g_k}{dx^{k}}\right)=g_0-\frac{d^mg_m}{dx^m}=f.\] Therefore the elements $(-1)^{k-1}(k-1)!\hat{f}^k$ are squarefree and satisfy \eqref{eq:hermite-list-condition}, so they agree with the $f_k\in\overline\K(x)$ satisfying \eqref{eq:hermite-list-output-k}.
\end{proof}

\begin{rem}\label{rem:hermite-list} As we mentioned in the introduction, we do not expect Algorithm~\ref{alg:hermite-list} to be surprising to the experts. What is surprising to us is that this trick is not used more widely since being originally suggested in \cite[\S5]{Horowitz1971}. We expect the theoretical cost of computing $\mathtt{HermiteList}(f)$ iteratively as in Algorithm~\ref{alg:hermite-list} is essentially the same as that of computing $\mathtt{HermiteReduction}(f)$ only once. This might seem counterintuitive, since the former is defined by applying the latter several times. But the size of the successive inputs in the loop decreases so quickly that the cost of the first step essentially dominates the added cost of the remaining steps put together. This conclusion is already drawn in \cite[\S5]{Horowitz1971} regarding the computational cost of computing iterated integrals of rational functions.
\end{rem}

\section{Simple Reduction}\label{sec:simple-reduction}

The results of the previous section allow us to further restrict our attention to proper rational functions $f\in\Kx$ with simple poles, which we write uniquely as $f=\frac{a}{b}$ with $a,b\in\K[x]$ such that $b$ is monic and squarefree, and either $a=0$ or else $0\leq\mathrm{deg}(a)<\mathrm{deg}(b)$.

Our next task is to compute a reduced form $\bar{f}\in\Kx$ such that $f-\bar{f}$ is rationally summable and $\bar{f}$ has squarefree denominator as well as polar dispersion $0$, which we accomplish in Algorithm~\ref{alg:simple-reduction}. As we mentioned already in the introduction, very many algorithms have been developed beginning with \cite{Abramov:1971} that can compute such a reduced form, even without assuming $f$ has simple poles.

Algorithm~\ref{alg:simple-reduction} requires the computation of the following set of integers, originally defined in \cite{Abramov:1971}.

\begin{defn} \label{def:shift-set}
    For $0\neq b\in\K[x]$, the \emph{(forward) shift set} of $b$ is \[\mathtt{ShiftSet}(b)=\{\ell\in\N \ | \ \mathrm{deg}(\mathrm{gcd}(b,\sigma^\ell(b)))\geq 1\}.\]
\end{defn}

 The following Algorithm~\ref{alg:shift-set} for computing $\mathtt{ShiftSet}(b)$ is based on the observation already made in \cite[p.~326]{Abramov:1971}, but with minor modifications to optimize the computations. 

\begin{algorithm}
\caption{$\mathtt{ShiftSet}$ procedure}\label{alg:shift-set}
\begin{algorithmic}
\Require  A polynomial $0\neq b\in\K$.
\Ensure  $\mathtt{ShiftSet}(b)$.\\
\If{$\mathrm{deg}(b) \leq 1$} {$S\gets\emptyset$;}
\Else \State{$R(z) \gets \mathrm{Resultant}_x(b(x), b(x+z))$;}\\
\State $\tilde{R}(z)\gets \dfrac{R(z)}{z\cdot \mathrm{gcd}_z\left(R(z), \frac{dR}{dz}(z)\right)};$ \Comment{Exact division.} \\
\State $T(z)\gets \tilde{R}(z^{\frac{1}{2}})$; \Comment{$\tilde{R}(z)$ is an even polynomial.}
\State{$S \gets \{\ell\in\mathbb{N} \ | \ T(\ell^2)=0\}$;}
\EndIf;\\
\Return $S$.
\end{algorithmic}
\end{algorithm}

\begin{lem}\label{lem:shift-set}
    Algorithm~\ref{alg:shift-set} is correct.
\end{lem}

\begin{proof}
    As pointed out in \cite[p.~326]{Abramov:1971}, $\mathtt{ShiftSet}(b)$ is the set of positive integer roots of the resultant $R(z)\in\K[z]$ defined in Algorithm~\ref{alg:shift-set}, which is the same as the set of positive integer roots of the square-free part $R(z)/\mathrm{gcd}_z\left(R(z),\frac{dR}{dz}(z)\right)$. It is clear that $R(0)=0$, and since we do not care for this root, we are now looking for positive integer roots of the polynomial $\tilde{R}(z)$ defined in Algorithm~\ref{alg:shift-set}. It follows from the definition of $R(z)$ that $R(\ell)=0$ if and only if $R(-\ell)=0$ for every $\ell\in\overline\K$ (not just for $\ell\in\Z$), and we see that this property is inherited by $\tilde{R}(z)$. Since $z\nmid \tilde{R}(z)$, the even polynomial $\tilde{R}(z)=T(z^2)$ for a unique $T(z)\in \K[z]$. 
\end{proof}

\begin{rem} \label{rem:shift-set}
    The role of the assumption that $\mathbb{K}$ be \emph{canonical} (cf.~\S\ref{sec:notation}) is made only so that one can count on a procedure to compute $\mathtt{ShiftSet}(b)$. We note that in \cite[\S6]{Gerhard:2003} a much more efficient (and general) algorithm than Algorithm~\ref{alg:shift-set} is described, which works for $b\in\Z[x]$. We remark that in order to compute $\mathtt{ShiftSet}(b)$ in general, it is sufficient to be able to compute a basis $\{w_1,\dots,w_s\}$ of the $\mathbb{Q}$-vector subspace of $\K$ spanned by the coefficients of the auxiliary polynomial $T(z)\in\K[x]$ defined in Algorithm~\ref{alg:shift-set}. Indeed, we could then write $T(z)=\sum_{j=1}^s v_jT_j(z)$ with each $T_j(z)\in\mathbb{Q}[z]$ and simply compute the set of (square) integer roots of $\tilde{T}(z)=\mathrm{gcd}(T_1,\dots,T_s),$ which is the same as the set of (square) integer roots of $T(z)$.
\end{rem}

The previous Algorithm~\ref{alg:shift-set} to compute $\mathtt{ShiftSet}$ is called by the following Algorithm~\ref{alg:simple-reduction} to compute reduced forms of rational functions with squarefree denominators.

\begin{algorithm}
\caption{$\mathtt{SimpleReduction}$ procedure}\label{alg:simple-reduction}
\begin{algorithmic}
\Require  A proper rational function $f\in\Kx$ with squarefree denominator.
\Ensure  A proper rational function $\bar{f}\in\Kx$ with squarefree denominator, such that $f-\bar{f}$ is rationally summable and either $\bar{f}=0$ or $\mathrm{pdisp}(\bar{f})=0$.\\
\State $b\gets \mathrm{denom}(f)$;
\State $S\gets\mathrm{ShiftSet}(b)$;
\If{$S =\emptyset $} \State {$\bar{f}\gets f$;}
\Else \For{$\ell\in S$}
\State $g_\ell\gets \mathrm{gcd}(b,\sigma^{-\ell}(b))$;
\EndFor;
\State $G\gets  \mathrm{lcm}(g_\ell \ | \ \ell\in S)$;
\State $b_0\gets \frac{b}{G}$; \Comment{Exact division.} 
 \For{ $\ell\in S$} 
 \State $b_\ell \gets gcd( \sigma^{-\ell}(b_0),b)$; 
\EndFor;
\State{$N\gets\{0\}\cup\{\ell\in S \ | \ \mathrm{deg}(b_\ell)\geq 1\}$};
\State $(a_\ell \ | \ \ell\in N) \gets \mathtt{ParFrac}(f;b_\ell \ | \ \ell\in N);$
\State{$\bar{f}\gets \displaystyle\sum_{\ell\in N}\sigma^\ell\left(\frac{a_\ell}{b_\ell}\right);$}
\EndIf;\\
\textbf{return} $\bar{f}$. 
\end{algorithmic}
\end{algorithm}

\begin{prop}\label{prop:simple-reduction}
Algorithm~\ref{alg:simple-reduction} is correct.
\end{prop}

\begin{proof}
    As in Algorithm~\ref{alg:simple-reduction}, let $b$ denote the denominator of $f$ and let $S=\mathtt{ShiftSet}(b)$. Then indeed if $S=\emptyset$ $\mathrm{pdisp}(f)=0$ so $f$ is already reduced and there is nothing to do. Assume from now on that $S\neq\emptyset$, and let us consider roots of polynomials in $\overline\K$. For each $\ell\in S$, the roots of $g_\ell:=\mathrm{gcd}(b,\sigma^{-\ell}(b))$ are those roots $\alpha$ of $b$ such that $\alpha-\ell$ is also a root of $b$. Therefore the roots of $G=\mathrm{lcm}(g_\ell:\ell\in S)$ are those roots $\alpha$ of $b$ such that $\alpha-\ell$ is also a root of $b$ for some $\ell\in\N$ (because all possible such $\ell$ belong to $S$, by the definition of $S$). It follows that the roots of $b_0:=b/G$ are those roots $\alpha$ of $b$ such that $\alpha-\ell$ is \emph{not} a root of $b$ for any $\ell\in\N$. In particular, $\mathrm{disp}(b_0)=0$. We call $b_0$ the \emph{divisor of initial roots}.
    
    Now the roots of $b_\ell:=\mathrm{gcd}(\sigma^{-\ell}(b_0),b)$ are those roots $\alpha$ of $b$ such that $\alpha-\ell$ is a root of $b_0$, i.e., the roots of $b$ which are precisely $\ell$ shifts away from the initial root in their respective $\mathbb{Z}$-orbits. It may happen that $b_\ell=1$ for some $\ell\in S$, because even though each $\ell\in S$ is the difference between two roots of $b$, it might be that no such pair of roots of $b$ involves any initial roots of $b_0$. Writing $N:=\{0\}\cup \{\ell\in S \ | \ \mathrm{deg}(b_\ell)=1\}$, it is clear that \begin{equation}\label{eq:simple-reduction-factorization}\prod_{\ell\in N}b_\ell=b\qquad \text{and} \qquad \mathrm{gcd}(b_\ell,b_j)=1 \ \text{for}\ \ell\neq j.\end{equation} Therefore we may uniquely decompose $f$ into partial fractions as in \eqref{eq:parfrac-def} with respect to the factorization \eqref{eq:simple-reduction-factorization} as called by Algorithm~\ref{alg:simple-reduction} \[f:=\sum_{\ell\in N} \frac{a_\ell}{b_\ell}\qquad \text{and set} \qquad \bar{f}:= \displaystyle\sum_{\ell\in N}\sigma^\ell\left(\frac{a_\ell}{b_\ell}\right).\]
    
    Now this $\bar{f}$ is a sum of proper rational functions with squarefree denominators, whence $\bar{f}$ also is proper with squarefree denominator.
    Since $\sigma(b_\ell)=\mathrm{gcd}(b_0,\sigma^\ell(b))$ is a factor of $b_0$ for each $\ell\in N$ and $\mathrm{disp}(b_0)=0$, we conclude that $\mathrm{pdisp}(\bar{f})=0$. Finally, for each $\ell\in N-\{0\}$ we see that \[\sigma^\ell\left(\frac{a_\ell}{b_\ell}\right)-\frac{a_\ell}{b_\ell}=\sum_{i=0}^{\ell-1}\sigma^i\left(\sigma\left(\frac{a_\ell}{b_\ell}\right)-\frac{a_\ell}{b_\ell}\right),\] whence $f-\bar{f}$ is a sum of rationally summable elements, and is therefore itself rationally summable.
\end{proof}

\begin{rem}\label{rem:simple-reduction}
  As we stated in the introduction, Algorithm~\ref{alg:simple-reduction} strikes us as being conceptually similar to the one already developed in \cite[\S5]{Gerhard:2003}, but its description is made simpler by our restriction to rational functions with simple poles only. Having a procedure that is easier for humans to read is not necessarily a computational virtue. But it is so in this case, because the relative simplicity of Algoritmhm~\ref{alg:simple-reduction} makes it also nimble and adaptable, enabling us in \S\ref{sec:extensions} to easily modify it to address other related applications beyond summability.  
\end{rem}

\section{Computation of discrete residues}

Now we wish to put together the algorithms presented in the earlier sections to compute symbolically the all the discrete residues of an arbitrary proper $f\in\Kx$, in the sense described in \S\ref{sec:dres-sum}. In order to do this, we first recall the following result describing the sense in which we compute classical residues symbolically by means of an auxiliary polynomial, and its short proof which explains how to actually compute this polynomial in practice.

\begin{lem}[\protect{\cite[Lem.~5.1]{trager:1976}}]\label{lem:trager}
    Let $f=\frac{a}{b}\in\Kx$ such that $a,b\in\K[x]$ satisfy $a\neq 0$, $\mathrm{deg}(a)<\mathrm{deg}(b)$, $\mathrm{gcd}(a,b)=1$, and $b$ is squarefree. Then there exists a unique polynomial $0\neq r\in\K[x]$ such that $\mathrm{deg}(r)<\mathrm{deg}(b)$ and \[f=\sum_{\{\alpha\in\overline\K \ | \ b(\alpha)=0\}}\frac{r(\alpha)}{x-\alpha}.\]
\end{lem}

\begin{proof}
    Since the set of poles of $f$ is the set of roots of $b$ and they are all simple poles, we know that the first-order residue $c_1(\alpha)$ of $f$ at each $\alpha\in\overline\K$ such that $b(\alpha)=0$ satisfies $0\neq c_1(\alpha)=a(\alpha)/\frac{db}{dx}(\alpha).$ Using the extended Euclidean algorithm we find the unique $0\neq r$ in $\K[x]$ with $\mathrm{deg}(r)<\mathrm{deg}(b)$ such that $r\cdot\frac{d}{dx}(b)\equiv a \pmod{b}.\qedhere$
\end{proof}

For $f\in\Kx$ satisfying the hypotheses of Lemma~\ref{lem:trager}, we denote \begin{equation}\label{eq:first-residues-def}\mathtt{FirstResidues}(f):=(b,r),\end{equation} where $r,b\in\K[x]$ are also as in the notation of Lemma~\ref{lem:trager}. We also define $\mathtt{FirstResidues}(0):=(1,0)$, for convenience. With this, we can now describe the following simple Algorithm~\ref{alg:dres} to compute a symbolic representation of the discrete residues of $f$.

\begin{algorithm}
\caption{$\mathtt{DiscreteResidues}$ procedure}\label{alg:dres}
\begin{algorithmic}
\Require  A proper rational function $0\neq f\in\Kx$.
\Ensure A list $\bigl((B_1,D_1),\dots,(B_m,D_m) \bigr)$ of pairs $(B_k,D_k)\in\K[x]^2$ such that every non-zero discrete residue of $f$ is of order at most $m$ and, for each $k=1,\dots,m$:
\begin{enumerate}
\item either $(B_k,D_k)=(1,0)$ or else $D_k\neq 0$, $\mathrm{deg}(D_k)<\mathrm{deg}(B_k)$, $B_k$ is squarefree, and $\mathrm{disp}(B_k)=0$;
\item the set of roots of $B_k$ in $\overline\K$ contains precisely one representative from each $\omega\in\overline\K/\Z$ such that $\mathrm{dres}(f,\omega,k)\neq 0$; and
\item $D_k(\alpha)=\mathrm{dres}(f,\omega(\alpha),k))$ for each root $\alpha\in\overline\K$ of $B_k$. \end{enumerate}\\

\State $(f_1,\dots,f_m)\gets \mathtt{HermiteList}(f)$;
\For{$k=1..m$}
\State $\bar{f}_k\gets\mathtt{SimpleReduction}(f_k)$;
\State $(B_k,D_k)\gets \mathtt{FirstResidues}(\bar{f}_k)$;
\EndFor;\\
\textbf{return} $\bigl((B_1,D_1),\dots,(B_m,D_m) \bigr)$. 
\end{algorithmic}
\end{algorithm}

\begin{thm}
\label{thm:dres}
    Algorithm~\ref{alg:dres} is correct.
\end{thm}

\begin{proof}
    It follows from the correctness of Algorithm~\ref{alg:hermite-list} proved in Lemma~\ref{lem:hermite-list} that $f$ has no poles of order greater than $m$, whence by Definition~\ref{def:dres} every non-zero discrete residue of $f$ has order at most $m$. Consider now $\bar{f}_k:=\mathtt{SimpleReduction}(f_k)$, which by the correctness of Algorithm~\ref{alg:simple-reduction} is such that $f-\bar{f}$ is summable. We prove the correctness of Algorithm~\ref{alg:dres} for each $k=1,\dots,m$ depending on whether $\bar{f}_k=0$ or not.

    In case $\bar{f}_k=0$, Algorithm~\ref{alg:dres} produces $(B_k,D_k)=(1,0)$. In this case we also know that $f_k$ is summable, and therefore by \eqref{eq:hermite-list-dres-k} $\mathrm{dres}(f_k,\omega,1)=\mathrm{dres}(f,\omega,k)=0$ for every $\omega\in\overline\K/\Z$. Thus the output of Algorithm~\ref{alg:dres} is (vacuously) correct in this case because the constant polynomial $B_k=1$ has no roots.

    Suppose now that $\bar{f}_k\neq 0$. It follows from the definition of $(B_k,D_k):=\mathtt{FirstResidues}(\bar{f}_k)$ as in \eqref{eq:first-residues-def} that $B_k$ is the denominator of the proper rational function $\bar{f}_k$, and therefore $B_k$ is non-constant, squarefree, and has $\mathrm{disp}(B_k)=0$, by the correctness of Algorithm~\ref{alg:simple-reduction} proved in Lemma~\ref{prop:simple-reduction}. Let us denote by $\bar{c}_k(\alpha)$ the classical first order residue of $\bar{f}_k$ at each $\alpha\in\overline\K$ (note that $\bar{f}_k$ has only simple poles, so there are no other residues). We obtain from Lemma~\ref{lem:trager} that $D_k\neq 0$, $\mathrm{deg}(D_k)<\mathrm{deg}(B_k)$, and $D_k(\alpha)=\bar{c}_k(\alpha)$ for each root $\alpha$ of $B_k$. Since $\bar{f}_k$ has at most one pole in each orbit $\omega\in\overline\K/\Z$ (this is what $\mathrm{pdisp}(\bar{f}_k)=0$ means), it follows that $\bar{c}_k(\alpha)=\mathrm{dres}(\bar{f}_k,\omega(\alpha),1)$ for every $\alpha\in\overline\K$. To conclude, we observe that \[\mathrm{dres}(\bar{f}_k,\omega,1)=\mathrm{dres}(f_k,\omega,1)=\mathrm{dres}(f,\omega,k)\] for each $\omega\in\overline\K/\Z$; the first equality follows from the summability of $f_k-\bar{f}_k$, and the second equality is \eqref{eq:hermite-list-dres-k}.
\end{proof}

\begin{rem}\label{rem:dres-deficiency}
    As we mentioned in \S\ref{sec:dres-sum}, the knowledge of a reduced form $h$ for $f$ is morally ``the same'' as knowledge of the discrete residues of $f$. And yet, the output $\bigl((B_1,D_1),\dots,(B_m,D_m) \bigr)$ of Algorithm~\ref{alg:dres} has the following deficiency: it may happen that for some $j\neq k$, we have $\mathrm{dres}(f,\omega,k)\neq 0 \neq \mathrm{dres}(f,\omega,j)$, and yet the representatives $\alpha_j,\alpha_k\in\omega$ such that $B_j(\alpha_j)=0=B_k(\alpha_k)$ may be distinct, with $\alpha_j\neq\alpha_k$. In many applications, this is not an issue because summability problems decompose into parallel summability problems in each degree component, as we see from Proposition~\ref{prop:dres-sum}. Actually, the systematic exploitation of this particularity was the original motivation of Algorithm~\ref{alg:hermite-list} and remains its \emph{raison d'\^etre}. But it is still unsatisfying that the different $B_k$ associated to the same $f$ are not better coordinated, and this does become a more serious (no longer merely aesthetic) issue in further applications to creative telescoping, where the discrete residues of different degrees begin to interact. We explain how to address this problem in Remark~\ref{rem:dres-deficiency-fix}, when we have developed the requisite technology.
\end{rem}

\section{Extensions and applications}\label{sec:extensions}

In this section we collect some modifications to the procedures described in the previous sections to produce outputs that allow for more immediate comparison of discrete residues accross several rational functions and accross different orders. 

We begin with the parameterized summability problem \eqref{eq:multi-telescoper} described in the introduction. Let $\mathbf{f}=(f_1,\dots,f_n)\in\Kx^n$ be given, and suppose we wish to compute a $\mathbb{K}$-basis for \begin{equation}\label{eq:v-space}V(\mathbf{f}):=\left\{\mathbf{v}\in\K^n \ \middle| \ \mathbf{v}\cdot\mathbf{f} \ \text{is summable}\right\}.\end{equation} By Proposition~\ref{prop:dres-sum}, \vspace{-.1in}\begin{equation}\label{eq:v-space-test}\mathbf{v}=(v_1,\dots,v_n)\in V(\mathbf{f}) \quad\Longleftrightarrow\quad \sum_{i=1}^nv_i\cdot\mathrm{dres}(f_i,\omega,k)=0\vspace{-.05in}\end{equation} for every $\omega\in\overline\K/\Z$ and every $k\in\N$, which is a linear system that we will be able to solve for the unknown $\mathbf{v}$ as soon as we know how to write it down. If we apply Algorithm~\ref{alg:dres} to each $f_i$ we obtain \[\mathtt{DiscreteResidues}(f_i)=\bigl((B_{i,1},D_{i,1}),\dots,(B_{i,m_i},D_{i,m_i})),\] and we run into the horrendous bookkeeping problem of having to decide, for each fixed $k\leq\max\{m_i \ | \ i=1,\dots,n\}$, for which orbits $\omega\in\overline\K/\Z$ it might happen that we have several different $\alpha_i\in\omega$ which are roots of $B_{i,k}$ but are not equal to one another on the nose.

To address this kind of problem, we introduce in Algorithm~\ref{alg:multi-reduction} a generalization of Algorithm~\ref{alg:simple-reduction} that computes reduced forms for several $f_1,\dots,f_n$ compatibly, so that whenever $f_i$ and $f_j$ have non-zero residue of order a given $k$ at a given orbit $\omega$ if and only if $B_{i,k}$ and $B_{j,k}$ have a common root $\alpha\in\overline\K$ such that $\omega=\omega(\alpha)$. For this purpose, we may assume as in \S\ref{sec:simple-reduction} that the $f_i$ are proper and, thanks to Algorithm~\ref{alg:hermite-list}, that they all have squarefree denominators.

\begin{algorithm}
\caption{$\mathtt{SimpleReduction}^+$ procedure}\label{alg:multi-reduction}
\begin{algorithmic}
\Require  An $n$-tuple $(f_1,\dots,f_n)\in\Kx^n$ of proper rational functions with squarefree denominators.
\Ensure  An $n$-tuple $(\bar{f}_1,\dots,\bar{f}_n)\in\Kx^n$ of proper rational functions with squarefree denominators, such that: each $f_i-\bar{f}_i$ is rationally summable; either $\bar{f}_i=0$ or $\mathrm{pdisp}(\bar{f}_i)=0$ for each $i$; and for every $1\leq i,j\leq n$ such that $\mathrm{dres}(f_i,\omega,1)\neq 0 \neq \mathrm{dres}(f_j,\omega,1)$ we have that $\bar{f}_i$ and $\bar{f}_j$ share a common pole in $\omega$.\\
 
 \State $(b_1,\dots,b_n)\gets \bigl(\mathrm{denom}(f_1),\dots,\mathrm{denom}(f_n)\bigr)$;
\State $b\gets\mathrm{lcm}(b_1,\dots,b_n)$
\State $S\gets\mathrm{ShiftSet}(b)$;
\If{$S =\emptyset $} \State {$(\bar{f}_1,\dots,\bar{f}_n)\gets (f_1,\dots,f_n)$;}
\Else \For{$\ell\in S$}
\State $g_\ell\gets \mathrm{gcd}(b,\sigma^{-\ell}(b))$;
\EndFor;
\State $G\gets  \mathrm{lcm}(g_\ell \ | \ \ell\in S)$;
\State $b_0\gets \frac{b}{G}$; \Comment{Exact division.}
\For{ $i=1..n$}
 \For{ $\ell\in S\cup \{0\}$} 
 \State $b_{i,\ell} \gets \mathrm{gcd}( \sigma^{-\ell}(b_0),b_i)$; 
\EndFor;
\State{$N_i\gets\{0\}\cup\{\ell\in S \ | \ \mathrm{deg}(b_{i,\ell})\geq 1\}$};
\State $(a_{i,\ell} \ | \ \ell\in N_i) \gets \mathtt{ParFrac}(f_i;b_{i,\ell} \ | \ \ell\in N_i);$
\State{$\bar{f}_i\gets \displaystyle\sum_{\ell\in N_i}\sigma^\ell\left(\frac{a_{i,\ell}}{b_{i,\ell}}\right);$}
\EndFor;\EndIf;\\
\textbf{return} $(\bar{f}_1,\dots,\bar{f}_n)$. 
\end{algorithmic}
\end{algorithm}

\begin{cor}\label{cor:multi-reduction}
    Algorithm~\ref{alg:multi-reduction} is correct.
\end{cor}

\begin{proof}
The proof is very similar to that of Proposition~\ref{prop:simple-reduction}, so we only sketch the main points. The key difference is that now $b_0$ has been defined so that for each root $\alpha$ of $b_0$ and each root $\alpha_i$ of $b_i$ belonging to $\omega(\alpha)$ we have that $\alpha_i-\alpha\in\mathbb{Z}_{\geq 0}$. The roots of $b_{i,\ell}$ are precisely those roots of $b_i$ which are $\ell$ steps away from the unique root of $b_0$ that belongs to the same orbit. By construction, the denominator of each $\bar{f}_i$ is a factor of $b_0$, which has $\mathrm{disp}(b_0)=0$ as before. \end{proof}

\begin{rem}\label{rem:dres-deficiency-fix}
    Algorithm~\ref{alg:multi-reduction} also allows us to fix the deficiency discussed in Remark~\ref{rem:dres-deficiency}. For a non-zero proper $f\in\Kx$, let us define $(f_1,\dots,f_m):=\mathtt{HermiteList}(f)$ as in Algorithm~\ref{alg:dres}. If we now set $(\bar{f}_1,\dots,\bar{f}_m):=\mathtt{SimpleReduction}^+(f_1,\dots,f_m)$, instead of $\bar{f}_k:=\mathtt{SimpleReduction}(f_k)$ separately for $k=1,\dots,m$, we will no longer have the problem of the $B_k$ being incompatible.
\end{rem}

More generally, we can combine Algorithm~\ref{alg:multi-reduction} with the modification proposed in the above Remark~\ref{rem:dres-deficiency-fix} to compute super-compatible symbolic representations of the discrete residues $\mathrm{dres}(f_i,\omega,k)$ of several $f_1,\dots,f_n\in\Kx$ which are compatible across the different $f_i$ as well as across the different $k\in\N$. This will be done in Algorithm~\ref{alg:multi-dres}, after explaining the following small necessary modification to the $\mathtt{FirstResidues}$ procedure defined in \eqref{eq:first-residues-def}. For an $n$-tuple of proper rational functions $\mathbf{f}=(f_1,\dots,f_n)$ with squarefree denominators, suppose $\mathtt{FirstResidues}(f_i)=:(b_i,r_i)$ as in \eqref{eq:first-residues-def}, and let $b:=\mathrm{lcm}(b_1,\dots,b_n)$. Letting $a_i:=\mathrm{numer}(f_i)$ and $d_i:=\frac{b}{b_i}$, we see that $\mathrm{gcd}(b_i,d_i)=1$ because $b$ is squarefree, and therefore by the Chinese Remainder Theorem we can find a unique $p_i\in\K[x]$ with $\mathrm{deg}(p_i)<b$ such that \vspace{-.05in} \[p_i\cdot\frac{d}{dx}(b_i)\equiv a_i \pmod{ b_i }\qquad \text{and} \qquad p_i\equiv 0 \pmod{d_i}.\] Then we see that $p_i(\alpha)$ is the first-order residue of $f_i$ at each root $\alpha$ of $b$. We define \vspace{-.03in}\[\mathtt{FirstResidues}^+(\mathbf{f}):=(b;(p_1,\dots,p_n)).\vspace{-.05in}\] 

\begin{algorithm}
\caption{$\mathtt{DiscreteResidues}^+$ procedure}\label{alg:multi-dres}
\begin{algorithmic}
\Require  An $n$-tuple $(f_1,\dots,f_n)\in\Kx^n$ of proper non-zero rational functions.
\Ensure A pair $(B;\mathbf{D})$, consisting of a polynomial $B\in\K[x]$ with $\mathrm{disp}(B)=0$ and an array $\mathbf{D}=\bigl(D_{i,k} \ | \ 1\leq i \leq n; 1\leq k \leq m \bigr)$ of polynomials $D_{i,k}\in\K[x]$, such that:
\begin{enumerate}
\item $B$ is non-constant and squarefree, and $\mathrm{disp}(B)=0$;
\item the set of roots of $B$ in $\overline\K$ contains precisely one representative from each $\omega\in\overline\K/\Z$ such that $\mathrm{dres}(f_i,\omega,k)\neq 0$ for some $i,k$; and
\item $D_{i,k}(\alpha)=\mathrm{dres}(f_i,\omega(\alpha),k))$ for each root $\alpha\in\overline\K$ of $B$. \end{enumerate}\\

\For{$i=1..n$}
\State $(f_{i,1},\dots,f_{i,m_i})\gets \mathtt{HermiteList}(f_i)$;
\EndFor;
\State $m\gets \mathrm{max}\{m_1,\dots,m_n\}$;
\For{$i=1..n$}
\For{$k=1..m$}
\If{$k>m_i$}
\State $f_{i,k}\gets 0$;
\EndIf;
\EndFor;\EndFor;
\State $\mathbf{f}\gets (f_{i,k} \ | \ 1\leq i \leq n; 1\leq k \leq m)$
\State $\bar{\mathbf{f}}\gets\mathtt{SimpleReduction}^+(\mathbf{f})$;\\
\textbf{return} $\mathtt{FirstResidues}^+(\bar{\mathbf{f}})$. 
\end{algorithmic}
\end{algorithm}

\begin{cor}\label{cor:multi-dres}
    Algorithm~\ref{alg:multi-dres} is correct.
\end{cor}
\begin{proof}
    This is an immediate consequence of the correctness of the Algorithm~\ref{alg:multi-reduction} for the $\mathtt{SimpleReduction}^+$ procedure, coupled with the same proof, \emph{mutatis mutandis}, given for Theorem~\ref{thm:dres}.
\end{proof}

Algorithm~\ref{alg:multi-dres} leads immediately to a simple algorithmic solution of the problem of computing $V(\mathbf{f})$ in \eqref{eq:v-space}.

\begin{prop}\label{prop:v-space}
    Let $\mathbf{f}=(f_1,\dots,f_n)$ with each $0\neq f_i\in\Kx$ proper. Let  \[\mathtt{DiscreteResidues}^+(\mathbf{f})=(B,\mathbf{D})\] with \[\mathbf{D}=(D_{i,k} \ | \ 1\leq i\leq n; \ 1\leq k\leq m)\] be as in Algorithm~\ref{alg:multi-dres} and let $V(\mathbf{f})$ be as in \eqref{eq:v-space}. Then \begin{equation}\label{eq:v-space-proof}V(\mathbf{f})=\left\{\mathbf{v}\in\K^n \ \middle| \ \sum_{i=1}^n v_iD_{i,k}=0 \ \text{for each} \ 1\leq k \leq m\right\}.\end{equation}
\end{prop}

\begin{proof} For each $\mathbf{v}\in\K^n$ and $\alpha\in\overline{\K}$ such that $B(\alpha)=0$, \begin{equation}\label{eq:dres-application}\mathrm{dres}(\mathbf{v}\cdot\mathbf{f},\omega(\alpha),k)=\sum_{i=1}^nv_iD_{i,k}(\alpha)\end{equation} by the correctness of Algorithm~\ref{alg:multi-dres} proved in Corollary~\ref{cor:multi-dres}. Moreover, since $\mathrm{deg}(D_{i,k})<\mathrm{deg}(B)$ and $B$ is squarefree, for each given $1\leq k\leq m$ we see that \eqref{eq:dres-application} holds for every root $\alpha$ of $B$ if and only if the polynomial $\sum_iv_iD_{i,k}=0$ identically. We conclude by Proposition~\ref{prop:dres-sum}.\end{proof}

\begin{rem}\label{rem:telescoping}
One can produce without too much additional effort a variant of Proposition~\ref{prop:v-space} that computes more generally the $\K$-vector space of solutions to the creative telescoping problem obtained by replacing the unknown coefficients $v_i\in\K$ in \eqref{eq:multi-telescoper} with unknown linear differential operators $\mathcal{L}_i\in\K\bigl[\frac{d}{dx}\bigr]$.
\end{rem}

\section{Connections with Galois theory}

We are interested in computing the vector space $V(\mathbf{f})$ in \eqref{eq:v-space} because it is isomorphic to the difference Galois group of the block-diagonal difference system $\sigma(Y)=(A_1\oplus\dots\oplus A_n)Y$ with diagonal blocks $A_i=\left(\begin{smallmatrix}1 & f_i \\ 0 & 1\end{smallmatrix}\right)$. Indeed, this difference Galois group consists of block-diagonal matrices $G(\mathbf{v})=G(v_1)\oplus\dots\oplus G(v_n)$ with diagonal blocks $G(v_i):=\left(\begin{smallmatrix}1 & v_i \\ 0 & 1\end{smallmatrix}\right)$ for $\mathbf{v}=(v_1,\dots,v_n)\in V(\mathbf{f})$.

Another application of the procedures developed here to Galois theory of difference equations arises from the consideration of diagonal systems \begin{equation}\label{eq:diagonal}
    \sigma(Y)=
    \mathrm{diag}(r_1,\dots,r_n)Y, \qquad \text{where} \ r_1,\dots,r_n\in\Kx^\times
\end{equation} As shown in \cite[\S2.2]{vanderput-singer:1997}, the difference Galois group of \eqref{eq:diagonal} is \[\Gamma:=\left\{(\gamma_1,\dots,\gamma_n)\in(\overline\K^\times)^n \ \middle| \ \gamma_1^{e_1}\cdots\gamma_n^{e_n}=1 \ \text{for} \ \mathbf{e}\in E\right\},\] where $E\subseteq \Z^n$ is the subgroup of $\mathbf{e}=(e_1,\dots,e_n)$ such that \begin{equation}\label{eq:multiplicative-relation}r_1^{e_1}\cdots r_n^{e_n}=\frac{\sigma(p_\mathbf{e})}{p_\mathbf{e}}\end{equation} for some $p_\mathbf{e}\in\Kx^\times$. 

Now suppose \eqref{eq:multiplicative-relation} holds, and let us write $f_i:=\frac{d}{dx}(r_i)r_i^{-1}$ for each $i=1,\dots,n$ and $g_\mathbf{e}:=\frac{d}{dx}(p_\mathbf{e})p_\mathbf{e}^{-1}$, so that we have \begin{equation}\label{eq:special-multi-telescoper}e_1f_1+\dots+e_nf_n=\sigma(g_\mathbf{e})-g_\mathbf{e}.\end{equation} At first glance, this looks like a version of problem \eqref{eq:multi-telescoper}, but it is even more special because the $f_i$ have only first-order residues, all in $\Z$. 
So we can compute the $\mathbb{Q}$-vector space $V$ of solutions to \eqref{eq:special-multi-telescoper} using $\mathtt{SimpleReduction}^+(f_1,\dots,f_n)$, just as in Proposition~\ref{prop:v-space}, as a preliminary step, and then compute a $\mathbb{Z}$-basis $\mathbf{e}_1,\dots,\mathbf{e}_s$ of the free abelian group $\tilde{E}:= V\cap\Z^n$. Since each $g_{\mathbf{e}_j}$ has only simple poles with integer residues, one can compute explicitly $p_{\mathbf{e}_j}\in\Kx$ such that $\frac{d}{dx}p_{\mathbf{e}_j}=g_{\mathbf{e}_j}p_{\mathbf{e}_j}$, and thence constants $\gamma_{j}\in\K^\times$ such that $\mathbf{r}^{\mathbf{e}_j}=\gamma_j$. This reduces the computation of $E$ from the defining multiplicative condition \eqref{eq:multiplicative-relation} in $\Kx^\times$ modulo the subgroup $\{\sigma(p) / p \ | \ p\in\Kx^\times\}$ to the equivalent defining condition in $\K^\times$: \[E=\left\{\sum_{j=1}^s m_j\mathbf{e}_j \in \tilde{E}\ \middle| \ \prod_{j=1}^s\gamma_j^{m_j}=1\right\}.\]

\section{Example} Let us conclude by illustrating some of our procedures on the following example considered in both \cite{Pirastu1995b,Malm:1995}. In order to make the computations easier for the human reader to follow, we have allowed ourselves to write down explicitly in this small example both irreducible factorizations of denominators.  
We emphasize and insist upon the fact that none of our procedures uses these factorizations.
    
Consider the rational function 
  \begin{equation}\label{eq:f-factorization}f:=\frac{1}{x^3(x + 2)^3(x + 3)(x^2 + 1)(x^2 + 4x + 5)^2}.  \end{equation}
We first compute $\mathtt{HermiteList}(f)=(f_1,f_2,f_3)$ with
\begin{gather*}
    f_1= \frac{787x^5 + 4803x^4 + 9659x^3 + 9721x^2 + 9502x + 5008}{18000(x^2 + 1)(x + 3)(x^2 + 4x + 5)(x + 2)x};\\
f_2=-\frac{787x^3 + 3372x^2 + 4696x + 1030}{18000(x^2 + 4x + 5)x(x + 2)};\quad \text{and}\\
f_3:=-\frac{7x - 1}{300(x + 2)x};
    \end{gather*} using Algorithm~\ref{alg:hermite-list}. We apply the remaining procedures to $f_1$ only -- the remaining $f_2$ and $f_3$ are similar and easier. Denoting \begin{equation}\label{eq:b-factorization} \notag b:=(x^2 + 1)(x + 3)(x^2 + 4x + 5)(x + 2)x\end{equation} the monic denominator of $f_1$, we compute with Algorithm~\ref{alg:shift-set} (or see by inspection) that $\mathtt{ShiftSet}(b)=\{1,2,3\}$. The factorization $b=b_0b_1b_2b_3$ computed within Algorithm~\ref{alg:simple-reduction} is given by
\begin{align*}
    b_0&=(x + 3)(x^2 + 4x + 5);\\
    b_1 &=\mathrm{gcd}(\sigma^{-1}(b_0),b)=x+2;\\
    b_2 & =\mathrm{gcd}(\sigma^{-2}(b_0),b)= x^2+1; \\
    b_3 &=\mathrm{gcd}(\sigma^{-3}(b_0),b)=x.
\end{align*}
The individual summands in the partial fraction decomposition of $f_1$ with respect to this factorization are given by
\begin{gather*}
    \frac{a_0}{b_0}=-\frac{13391x^2 + 37742x - 9293}{1080000(x + 3)(x^2 + 4x + 5)};\\
    \frac{a_1}{b_1} =\frac{1}{250(x+2)};\qquad
    \frac{a_2}{b_2} = \frac{-7x-1}{800(x^2+1)}; \qquad
    \frac{a_3}{b_3} =\frac{313}{33750x}.
\end{gather*}
The reduced form $\bar{f}_1=b_0+\sigma\left(\frac{a_1}{b_1}\right)+\sigma^2\left(\frac{a_2}{b_2}\right)+ \sigma^3\left(\frac{a_3}{b_3}\right)$ is given by \[\bar{f}_1=\frac{273x + 1387}{20000(x + 3)(x^2 + 4x + 5)}.\] The first pair of polynomials $(B_1,D_1)\in\K[x]$ computed by Algorithm~\ref{alg:dres} is given by \[B_1=(x+3)(x^2+4x+5) \qquad \!\text{and} \!\qquad D_1=\frac{59}{16000}x^2+\frac{33}{40000}x-\frac{1321}{80000}.\]

Let us compare this output $(B_1,D_1)$ with the discrete residues $\mathrm{dres}(f,\omega,1)$ of $f$ according to the Definition~\ref{def:dres} in terms of classical residues. We see from the factorization of the denominator of $f$ in \eqref{eq:f-factorization} that its set of poles is \[\bigl\{-3,-2,0,\sqrt{-1}-2, \sqrt{-1},-\sqrt{-1}, -\sqrt{-1}-2\bigr\},\] 
and that each of these poles belongs to one of the three orbits \[\omega(0)=\mathbb{Z}\qquad\text{and}\qquad\omega\bigl(\pm \sqrt{-1}\bigr)=\pm\sqrt{-1}+\Z.\] Therefore, $f$ has no discrete residues outside of these orbits, and we verify that the set of roots $\{-3,\sqrt{-1}-2,-\sqrt{-1}-2\}$ of the polynomial $B_1$ correctly contains precisely one representative from each of these orbits (subject to our verification below that the first-order discrete residues of $f$ at these orbits are actually non-zero!). We can compute directly in this small example that the classical first-order residues of $f$ at the poles in $\omega(0)$ are given by 
\[
    \mathrm{Res}_1(f,0)=\tfrac{313}{33750};\quad \mathrm{Res}_1(f,-2)=\tfrac{1}{250}; \quad \mathrm{Res}_1(f,-3)=\tfrac{1}{1080};
\] and at the poles in $\omega\bigl(\pm\sqrt{-1}\bigr)$ are given by \vspace{-.1in}\begin{gather*}\mathrm{Res}_1\bigl(f,\pm\sqrt{-1}\bigr)=\frac{\pm\sqrt{-1}-7}{16000};\quad\text{and}\\ \mathrm{Res}_1\bigl(f,\pm\sqrt{-1}-2\bigr)=\frac{\mp1119\sqrt{-1}-533}{80000}.\end{gather*}
Finally we can verify directly that the polynomial $D_1$ correctly computes the first-order discrete residue of $f$ at all three orbits $\omega(0)$ and $\omega\left(\pm\sqrt{-1}\right)$ according to Definition~\ref{def:dres}: \[\mathrm{dres}(f,\omega(0),1)=\tfrac{313}{33750}+\tfrac{1}{250}+\tfrac{1}{1080}=\tfrac{71}{5000}=D_1(-3);\]and\begin{multline*}
\mathrm{dres}\bigl(f,\omega\bigl(\pm\sqrt{-1}\bigr),1\bigr)=\frac{\pm\sqrt{-1}-7}{16000}+\frac{\mp1119\sqrt{-1}-533}{80000}=\\ =\frac{\mp557\sqrt{-1}-284}{40000}=D\bigl(\pm\sqrt{-1}-2\bigr).\end{multline*}

\section*{Acknowledgements}

 Both authors gratefully acknowledge the support of NSF grant CCF-1815108 and a UTD startup grant provided to the first author. We also thank Shaoshi Chen and Michael Singer for helpful discussions and suggestions during the preparation of the second author's PhD thesis \cite{sitaula:2023}, on which this manuscript is based.

\end{document}